\documentclass[10pt, twocolumn, conference]{IEEEtran}
\usepackage{times}
\usepackage[reqno]{amsmath}
\usepackage{amsfonts}
\usepackage{caption3}
\usepackage{times,amsmath,epsfig}
\usepackage{latexsym,amssymb}
\usepackage{rotating,color}
\usepackage{cite}
\usepackage{extarrows}
\usepackage{amsthm}
\usepackage[dvipsnames]{pstricks}
\usepackage{algpseudocode}
\usepackage{algorithm}
\usepackage{url}
\usepackage{comment}
\usepackage{mathtools}
\usepackage{bbm}
\usepackage{kbordermatrix}
\usepackage{lipsum}
\usepackage{multicol}
\DeclarePairedDelimiter\ceil{\lceil}{\rceil}
\DeclarePairedDelimiter\floor{\lfloor}{\rfloor}

 \newcommand{\no}{\nonumber}

\newcommand{\bA}{\mathbf{A}}
\newcommand{\bB}{\mathbf{B}}
\newcommand{\bH}{\mathbf{H}}
\newcommand{\bM}{\mathbf{M}}

\newcommand{\bg}{\mathbf{g}}
\newcommand{\bx}{\mathbf{x}}
\newcommand{\by}{\mathbf{y}}

\newcommand{\bz}{\boldsymbol{0}}
\newcommand{\bv}{\boldsymbol{v}}
\newcommand{\bw}{\boldsymbol{w}}
\newcommand{\bu}{\boldsymbol{u}}
\newcommand{\cS}{\mathcal{S}}
\newcommand{\cT}{\mathcal{T}}

\newcommand{\bI}{\mathbf{I}}
\newtheorem{thm}{Theorem}

\newtheorem{lem}{Lemma}

\newtheorem{define}{Definition}

\IEEEoverridecommandlockouts

\DeclarePairedDelimiter\dpro{\langle}{\rangle}
\usepackage{blkarray}
\makeatletter
\def\ps@headings{%
\def\@oddhead{\mbox{}\scriptsize\rightmark \hfil \thepage}%
\def\@evenhead{\scriptsize\thepage \hfil \leftmark\mbox{}}%
\def\@oddfoot{}%
\def\@evenfoot{}}
\makeatother
\IEEEoverridecommandlockouts

\allowdisplaybreaks

\begin{document}
\title{Precise Phase Transition of Total Variation Minimization}
\author{\IEEEauthorblockN{Bingwen Zhang}
\IEEEauthorblockA{Dept. of ECE\\
Worcester Polytechnic Institute\\
Worcester, MA 01609\\
Email: bzhang@wpi.edu}
\and
\IEEEauthorblockN{Weiyu Xu}
\IEEEauthorblockA{Dept. of ECE\\
University of Iowa\\
Iowa City, IA 52242\\
Email:weiyu-xu@uiowa.edu}
\and
\IEEEauthorblockN{Jian-Feng Cai}
\IEEEauthorblockA{Dept. of Mathematics\\
Hong Kong U. of Sci. \& Tech.\\
Clear Water Bay, Hong Kong\\
Email:jfcai@ust.hk}
\and
\IEEEauthorblockN{Lifeng Lai}
\IEEEauthorblockA{Dept. of ECE\\
Worcester Polytechnic Institute\\
Worcester, MA 01609\\
Email: llai@wpi.edu}}

%\author{Bingwen Zhang, Weiyu Xu and Lifeng Lai \thanks{Bingwen Zhang and Lifeng Lai are with Department of Electrical and Computer Engineering,
%Worcester Polytechnic Institute, Worcester, MA, 01609. Email:\{bzhang,llai\}@wpi.edu.\\ Weiyu Xu is with Department of Electrical and Computer Engineering, University of Iowa, Iowa City, IA, 52242. Email: weiyu-xu@uiowa.edu.}
%}
\maketitle %\pagestyle{plain}

%%%%%%%%%%%%%%% Article Body %%%%%%%%%%%%%%%%%%%%%%%%%%%%%%%%%%%%%%%%%

%%%%=======================================================
\begin{abstract}
Characterizing the phase transitions of convex optimizations in recovering structured signals or data is of central importance in compressed sensing, machine learning and statistics.
The phase transitions of many convex optimization signal recovery methods such as $\ell_1$ minimization and nuclear norm minimization are well understood through recent years' research.
However, rigorously characterizing the phase transition of total variation (TV) minimization in recovering sparse-gradient signal is still open.
In this paper, we fully characterize the phase transition curve of the TV minimization. Our proof builds on Donoho, Johnstone and Montanari's conjectured phase transition curve
for the TV approximate message passing algorithm (AMP), together with the linkage between the minmax Mean Square Error of a denoising problem and the high-dimensional convex geometry for TV minimization.
\end{abstract}
%\begin{keywords}
%Phase Transition; Total Variation Minimization; Gaussian width.
%\end{keywords}
\section{Introduction}
In the last decade, using convex optimization to recover parsimoniously-modeled signal or data from a limited number of samples has attracted significant research interests in
compressed sensing, machine learning and statistics \cite{candes2006robust,donoho2005neighborliness,candes2009exact,Donoho:TIT:2013}. For example, in compressed sensing, the main idea
is to exploit the sparse structures inherent to the underlying signal, and design sparsity-promoting convex optimization programs, such as $\ell_1$ minimization,  to efficiently recover the signal
from a much smaller number of measurements than the ambient signal dimension.  Numerical results empirically show that these convex optimization based signal recovery algorithms often exhibit
a phase transition phenomenon: when the number of measurements exceeds a certain threshold, the convex optimization can correctly recover the structured signals with high probability; when the
number of measurements is smaller than the  threshold, the convex optimization will fail to recover the underlying structured signals with high probability.
A series of works studying convex geometry for linear inverse problems have made substantial progress in theoretically characterizing the phase transition phenomenon for convex optimizations in
recovering structured signals \cite{donoho2005neighborliness, Mihalo:arxiv:2009,Venkat:arXiv:2012,Tropp:Info:2014,bayati2011dynamics}. For example, the phase transitions
for  $\ell_1$ minimization used in recovering sparse signals and nuclear norm minimization used in recovering low-rank matrix have been  well
understood \cite{donoho2005neighborliness,Mihalo:arxiv:2009,Venkat:arXiv:2012,Tropp:Info:2014,bayati2011dynamics}.

In spite of all this progress, characterizing the phase transition for the total variation minimization used in recovering sparse-gradient signals is still open.  Sparse-gradient signals are signals that are piece-wise constant, and thus have a small number of non-zero gradients. This type of signals arise naturally in applications in signal denoising and in digital image processing \cite{Tibshirani:JRSS:2005, Liu:KDD:2010, Leonid:Phys:1992}. Let $\bx^*\in\mathbb{R}^n$ be a vector representing a one-dimensional piece-wise constant signal, and $\bB\bx^*$ denote the finite difference of $\bx^*$, in which $(\bB\bx^*)_i=\bx^*_{i+1}-\bx^*_{i}$ with $\bx^*_{i}$ being the $i$th element of $\bx^*$. Since $\bx^*$ has sparse gradients, $\bB\bx^*$ has very few non-zero entries. Suppose one observes $\by=\bA\bx^*$, in which $\bA\in \mathbb{R}^{m\times n}$ is the observation matrix, then in the total variation (TV) minimization problems, one tries to recover $\bx^*$ from $\by$ by solving
\begin{eqnarray}
\min\limits_{\bx} &&\|\bB\bx\|_1,\label{eq:TVproblem}\\
\text{s.t.}&& \by=\bA\bx.\nonumber
\end{eqnarray}
Here, $\|\bB\bx\|_1=\sum\limits_{ i=1}^{n-1}(\bB\bx)_i$ is called the total variation semi-norm of $\bx$.% and $A\in \mathbb{R}^{n\times p}$ is the measurement matrix.

TV minimization has a wide range of applications, including image reconstruction and restoration\cite{Cai:JAMS:2012, Sidky:PMB:2008}, medical imaging\cite{Keeling:AMC:2008}, noise removing\cite{Leonid:Phys:1992}, computing surface evolution\cite{Chambolle:IJCV:2009} and profile reconstruction\cite{Berg:IP:1995}. However, the understanding of the performance of TV minimization is less complete than that of other convex optimization based methods such as $\ell_1$ minimization. In particular, the phase transition of the TV minimization has not been fully characterized and remains as an open problem. In this paper, we solve this open problem of fully characterizing the phase transition of the TV regularization. The starting points of our investigation are the results obtained in~\cite{Tropp:Info:2014} and~\cite{Donoho:TIT:2013}, which we discuss in detail in the following.

First, for a general signal recovery problem using general proper convex penalty function $f(\bx)$ given as follows,
\begin{eqnarray}
\min\limits_{\bx} &&f(\bx),\label{eq:GeneralRegularizer}\\
\text{s.t.}&& \by=\bA\bx,\nonumber
\end{eqnarray}
the authors of \cite{Tropp:Info:2014} showed that the phase transition on the number of measurements happens at the Gaussian width of the descent cone of the proper convex penalty function $f(\bx)$.
Using this result and earlier results from polyhedral geometry, researchers have fully characterized the phase transition thresholds for $\ell_1$ minimization and nuclear norm minimization by calculating
the Gaussian width of their decent cones. However, since the total variation semi-norm is a non-separable convex penalty term, calculating the precise Gaussian width of the descent cone of the total variation
semi-norm is difficult and remains open. This difficulty in calculating the Gaussian width also prevents us from characterizing the phase transition of total variation minimization in recovering sparse-gradient signals.

Second, in \cite{Donoho:TIT:2013}, the authors first considered a denoising problem where the total variation regularizer was used to denoise sparse-gradient signals contaminated by additive Gaussian noises, and characterized the minimax MSE of this denoising method.  The authors in \cite{Donoho:TIT:2013} further proposed an approximate message passing algorithm for recovering sparse-gradient measurements from undersampled measurements, and
conjectured that the minimax MSE for the denoising problem was the same as the phase transition (the number of measurements) for the approximate message passing algorithm.  Numerical results in
\cite{Donoho:TIT:2013} demonstrated that the empirical phase transitions for both the AMP algorithm and the total variation minimization~\eqref{eq:TVproblem}  match the minimax MSE for the denoising problem.
However, justifying the conjecture in \cite{Donoho:TIT:2013} requires the assumption that the state evolution for the approximate message passing algorithm is valid, which still remains to be proved.
Furthermore, we do not know whether the AMP and the total variation minimization indeed have the same phase transition. In \cite{Oymak:arxiv:2013}, the authors showed that the minimax MSE of the denoising problem considered in \cite{Donoho:TIT:2013}
is an upper bound on the phase transition (the number of needed measurements) of total variation minimization (as will be discussed later in this paper).
However, it remains unknown whether the minimax MSE of the denoising problem is indeed the phase transition of total variation minimization.

%In this paper, we solve this open problem of fully characterizing the phase transition of the TV regularization. The starting points of our investigation are the results obtained in~\cite{Donoho:TIT:2013} and \cite{Venkat:arXiv:2012}. In particular, among other results, \cite{Donoho:TIT:2013} connected the approximate message passing algorithm for recovering sparse-gradient signals~\eqref{eq:TVproblem} to a minimax denoising problem using TV regularization terms. It was conjectured in \cite{Donoho:TIT:2013} that the minimax MSE of the TV-regularized denoiser gives the phase transition for the approximate message passing algorithm and total variation minimization. Using the tools from convex geometry, \cite{Rudelson_onsparse}, \cite{Mihalo:arxiv:2009}, \cite{Venkat:arXiv:2012} and \cite{Tropp:Info:2014} provided the Gaussian width of the decent cone of a proper convex function $f(\bx)$ as the phase transition for signal recovery by minimizing $f(\bx)$ in \eqref{eq:GeneralRegularizer}. While the Gaussian width tool is applicable for the TV regularizer, it is not clear how to compute the Gaussian width for TV minimization.
As our main contribution in this paper, we rigorously prove that the minimax MSE of TV-regularized denoising considered by~\cite{Donoho:TIT:2013} is indeed the phase transition of the TV minimization problem~\eqref{eq:TVproblem}, by showing the minimax MSE of the denoising problem is approximately equal to the Gaussian width of the descent cone of the TV semi-norm, up to negligible constants.  We remark that, different from the Gaussian width, the minimax MSE of the TV-regularized denoising can be readily computed. We can thus characterize the phase transition of total variation minimization using the minimax MSE of the denoising problem.

Here, we would like to compare our work with~\cite{Jianfeng:Info:2015}. In \cite{Jianfeng:Info:2015}, the authors gave upper and lower bounds on the number of needed measurements for recovering  worst-case  sparse-gradient signals which have a fixed number of nonzero elements in its signal gradient, using the tool of Gaussian width. In contrast, in this paper we will focus on the phase transition for average-case sparse-gradient signals, where the number of nonzero elements in signal gradient grows proportionally with the ambient signal dimension.

The remainder of the paper is organized as follows. In Section \ref{sec:background}, we introduce the background and set up the notations that will be used in later analysis and proofs. In Section \ref{sec:result}, we verify that the TV regularizer satisfies the weak decomposability condition in \cite{Foygel:TIT:2014} and use this condition to fully characterize the phase transition of the TV minimization problem. In Section \ref{sec:conclusion}, we provide several concluding remarks.

\section{Background}\label{sec:background}
%We have observations
%\begin{equation}
%y = Ax^{*},
%\end{equation}
%where $x^{*}\in\mathbb{R}^p$ and $A$ is a $n\times p$ observation matrix. $n$ is the total number of observations and $p$ is the dimension of signal. We want to recover the signal $x^{*}$ from the observations $y$ by total variation(TV) minimization\cite{Leonid:Phys:1992} or fused Lasso\cite{Tibshirani:JRSS:2005}, which is of the form
%\begin{equation}
%\hat{x}=\underset{x\in\mathbb{R}^p}{\arg\min}\frac{1}{2}\|y-x\|_2^2+\tau f(x),\label{eq:TV}
%\end{equation}
%where $\tau\in\mathbb{R}_{+}$, $f(x)$ being the TV regularizer $f(x):=\sum_{i=1}^{p-1}|x_{i+1}-x_i|$ and $\hat{x}$ is an estimator of $x^{*}$.

\subsection{Definitions and Notations}
We first introduce definitions and notations that will be used throughout the paper.

We use $f(\bx)$ to denote the TV regularizer $f(\bx):=\|\bB\bx\|_1$, which is not a norm, and
$\bB\in\mathbb{R}^{(n-1)\times n}$ with
\begin{equation}
\bB_{i,j} =
\begin{cases}
1 &\mbox{if } j=i \\
-1 &\mbox{if } j=i+1 \\
0 &\mbox{otherwise}.
\end{cases}
\end{equation}
Let $\partial f(\bx)$ be the subdifferential of $f$ at $\bx$.

For a given non-empty set $\mathcal{C}\subseteq \mathbb{R}^n$, the cone obtained by $\mathcal{C}$ is defined as
\begin{eqnarray}
\text{cone}(\mathcal{C}):=\{\lambda \bx\in\mathbb{R}^n: \bx \in\mathcal{C},\lambda\geq 0\}.
\end{eqnarray}

The distance from a vector $\bg\in\mathbb{R}^n$ to the set $\mathcal{C}$ is defined as
\begin{eqnarray}
\text{dist}(\bg,\mathcal{C}):=\inf_{\bu\in \mathcal{C}}\|\bg-\bu\|_2,
\end{eqnarray}
in which $\|\cdot\|_2$ is the $\ell_2$ norm.

The mean square distance to $\mathcal{C}$ is defined as
\begin{eqnarray}
D(\mathcal{C}):= \mathbb{E}\{\text{dist}(\bg,\mathcal{C})^2\},
\end{eqnarray}
in which the expectation is taken over $\bg\sim \mathcal{N}(\bz,\bI)$ with $\bI$ being the identity matrix.
%The TV minimization is most often used in digital image processing and has applications in noise removal. In TV minimization problem, it is of interest to identify the phase transition\cite{Tropp:arxiv:2005}. Generally speaking, a phase transition is a sharp change of the results in problems. In TV minimization problem, a phase transition is the change for the number of samples $n$ from failure to success for fully recovering $x^{*}$ from $y$ or $\hat{x}=x^{*}$. In particular, we will like to study the phase-transition phenomena on $n$, the number of observation needed, for $\delta\in\mathbb{R}_{+}$, such that: a) if $n>p\delta$, \eqref{eq:TV} can successfully recover $x^*$, i.e., $\hat{x}=x^*$, with high probability; b) if $n<p\delta$, \eqref{eq:TV} fails to recover $x^*$, i.e., $\hat{x}\neq x^*$, with high probability. The open problem is that the exact value of $\delta$ is unknown for TV minimization.

%\noindent\textbf{Upperbound on $\delta$}

%\vspace{5mm}
%Let $f(x)$ denote the TV regularizer and $\partial f(x)$ be the subdifferential of $f$ at the underlying signal $x$.

Throughout the paper, we will use $[k]:=\{1,2,\cdots,k\}$ where $k$ is a positive integer, $[b,e]:=\{b,b+1,\cdots,e\}$ where $e\geq b$. Similarly, $(b,e):=\{b+1,b+2,\cdots,e-1\}$. Let $\mathcal{S}$ be a subset of $[n-1]$, then $\mathcal{S}^c$ denote the complement of $\mathcal{S}$ with respect to $[n-1]$. We will use $|\mathcal{S}|$ to denote the cardinality of the set $\mathcal{S}$.

Let $\bu\in\mathbb{R}^{n-1}$ be a vector and $\mathcal{S}$ be a subset of the indices set $[n-1]$, then $\bu_{\mathcal{S}}\in\mathbb{R}^{n-1}$ is the vector such that
\begin{equation}
(\bu_{\mathcal{S}})_{i}=
\begin{cases}
\bu_{i}, &\mbox{ if } i\in \mathcal{S}\\
0, &\mbox{ if } i\notin \mathcal{S}.
\end{cases}\label{eq:notation}
\end{equation}

We use $\tilde{\bu}_{\cS}\in\mathbb{R}^{|\cS|}$ to denote the shortened version of $\bu_{\cS}$ by deleting all zeros in $\bu_{\cS}$. To be more explicit, let $\cS=\{s_1,s_2,\cdots,s_{|\cS|}\}$,
\begin{equation}
(\tilde{\bu}_{\cS})_{i}=\bu_{s_i},\text{ }\forall i\in [|\cS|].\label{eq:subnotation}
\end{equation}

Let $\bM\in\mathbb{R}^{(n-1)\times(n-1)}$ be a matrix, and $\cS$ and $\cT$ be subsets of $[n-1]$, then $\bM_{\cS,\cT}\in\mathbb{R}^{|\cS|\times |\cT|}$ is the matrix produced by deleting all rows not in $\cS$ and columns not in $\cT$ from $\bM$. To be explicit, let $\cS=\{s_{1},s_{2},\cdots,s_{|\cS|}\}$ and $\cT=\{t_{1},t_{2},\cdots,t_{|\cT|}\}$,
\begin{equation}
(\bM_{\cS,\cT})_{i,j}=\bM_{s_i,t_j}, \text{ }\forall i\in [|\cS|]\text{ and }\forall j\in [|\cT|].\label{eq:matrixnotation}
\end{equation}
We also write ${\bM}_{\cS,\cT}$ as ${\bM}_{\cS,\Omega}$ if $\cT=[n-1]$. Similarly, if $\cS=[n-1]$, we write ${\bM}_{\cS,\cT}$ as ${\bM}_{\Omega, \cT}$.

\subsection{Phase Transition for the AMP~\cite{Donoho:TIT:2013}}

In \cite{Donoho:TIT:2013}, to recover sparse-gradient signals from undersampled measurements, the authors proposed an iterative approximate message passing algorithm, called TV-AMP algorithm,
which uses the TV denoisers in each iteration.  The authors further connected the TV-AMP algorithm with the minimax denoising problem.  In the denoising problem, one observes $\by=\bx^*+\mathbf{z}$, in which $\mathbf{z}$ is the noise vector with i.i.d. standard Gaussian random variables with unit variance,  and tries to recover $\bx^*$ from the noisy observation $\by$. In particular, \cite{Donoho:TIT:2013} conjectured that the minimax MSE of the denoising problem will correctly predict the phase transition of the TV-AMP algorithm.  Moreover, it is observed that the minimax MSE of the denoising problem matches the empirical phase transition of ~\eqref{eq:TVproblem}, and the empirical phase transition of the AMP algorithm.
Let $m_{MAP}$ be the number of observation needed for the AMP algorithm. \cite{Donoho:TIT:2013} numerically showed that, as soon as $m_{MAP}\geq nM_{\text{denoiser}}$, the AMP algorithm will be
successful in recovering $\bx^*$ with a high probability. Here $M_{\text{denoiser}}$ is the per-coordinate minimax mean squared error of the denoising problem when one observes $\by=\bx^*+\mathbf{z}$ and
uses the TV-penalized least-square denoisers.  However, in \cite{Donoho:TIT:2013}, the analytically derived phase transition for the AMP algorithms depends on the assumption of the AMP state evolution being correct.
However, proving that the assumption holds true remains open for the TV-AMP. Moreover, it is unknown whether the phase transition of the AMP algorithm theoretically matches the phase transition of the TV minimization ~\eqref{eq:TVproblem}.
Thus characterizing the phase transition for the TV minimization remains open, even though we have a phase transition formula from \cite{Donoho:TIT:2013} matching the empirical performance of TV minimization .

In another line of work using convex geometry, \cite{Oymak:arxiv:2013} showed that the minimax MSE $M_{\text{denoiser}}$ is closely related to $\min_{\lambda\geq 0}D(\lambda\partial f(\bx))$,
where $\partial f(\bx)$ is the subdifferential of $f(\bx)$ at the underlying signal $\bx$.
In particular,~\cite{Oymak:arxiv:2013} showed that $nM_{\text{denoiser}}\approx \min\limits_{\lambda\geq 0}D(\lambda\partial f(\bx))$.
However, it is still unknown whether $\min_{\lambda\geq 0}D(\lambda\partial f(\bx))$ provides the phase transition for the AMP or the TV minimization ~\eqref{eq:TVproblem}.

\subsection{Phase Transition Based on Gaussian Width Calculation~\cite{Venkat:arXiv:2012}}

Using the ``escape through the mesh''  lemma,  recent works \cite{Rudelson_onsparse, Mihalo:arxiv:2009, Venkat:arXiv:2012, Tropp:Info:2014} have shown that,
for a proper convex function $f(\cdot)$,  $D(\text{cone}(\partial f(\bx_0)))$ (where $\bx_0$ is the original signal) is the phase transition threshold on the number of needed Gaussian measurements
for the optimization problem
~\eqref{eq:GeneralRegularizer} to recover $\bx_0$. As discussed above, while this formula $D(\text{cone}(\partial f(\bx_0)))$ is applicable for the TV minimization problem, it is not clear how to compute it for the TV
semi-norm function $f(\bx)$, which is a non-separable function. This is in contrast to  the Gaussian width calcaulations for separable penalty functions such as $\ell_1$ norms.

\subsection{Central Issue and Our Approach}
At this point, it is not known whether $\min_{\lambda\geq 0}D(\lambda\partial f(\bx))\approx D(\text{cone}(\partial f(\bx)))$ or not for the TV regularizer $f(\bx)$.
Thus it is not clear whether the minmax MSE result derived in \cite{Donoho:TIT:2013} will directly give the phase transition of the TV minimization.
In fact,  when $f(\bx)$ represents a norm of $\bx$, it is known that $\min_{\lambda\geq 0}D(\lambda\partial f(\bx))\approx D(\text{cone}(\partial f(\bx)))$ \cite{Tropp:Info:2014}.
One may thus wonder whether we can show this equality to hold for the TV regularizer by directly applying (3.5) in \cite{Oymak:arxiv:2013} or (4.3) in \cite{Tropp:Info:2014}.
However, there are two obstacles for directly applying those two equations. First, the TV regularizer $f(\bx)$ is not a norm but a semi-norm instead.
Secondly, even if we go ahead with applying (3.5) in \cite{Oymak:arxiv:2013} or (4.3) in \cite{Tropp:Info:2014} to bound the Gaussian width of the descent cone of the function $f(\bx)$, the approximation error is too big, since $1/f(\bx/\|\bx\|_2)$
can be arbitrarily big for an $n$-dimensional signal $\bx$, when $f(\bx)$ is the total variation semi-norm.

In this paper, we will show that
\begin{eqnarray}
\min\limits_{\lambda\geq 0}D(\lambda\partial f(\bx))\approx D(\text{cone}(\partial f(\bx))),\label{eq:key}
\end{eqnarray}
for the TV regularizer, and $\min\limits_{\lambda\geq 0}D(\lambda\partial f(\bx))$ is indeed the phase transition of the TV regularizer.

In order to show~\eqref{eq:key}, we instead build on Proposition 1 of \cite{Foygel:TIT:2014}. In particular, we show that $f(\bx)$ satisfies the weak decomposability condition defined in \cite{Foygel:TIT:2014}, and
hence we can use Proposition 1 of \cite{Foygel:TIT:2014} to obtain:
\begin{eqnarray}
\min\limits_{\lambda\geq 0}D(\lambda\partial f(\bx))\leq D(\text{cone}(\partial f(\bx)))+6,
\end{eqnarray}
which coupled with the fact that $$ \min\limits_{\lambda\geq 0}D(\lambda\partial f(\bx))\geq D(\text{cone}(\partial f(\bx)))$$ proves~\eqref{eq:key}.

\section{Main Result}\label{sec:result}
In this section, we prove that $\min\limits_{\lambda\geq 0}D(\lambda\partial f(\bx))$ is the phase transition of~\eqref{eq:TVproblem} by showing that~\eqref{eq:key} holds.
%In this section, we verify the weak decomposability condition stated in Proposition 1 of \cite{Foygel:TIT:2014}.
For any given nonzero vector $\bx\in\mathbb{R}^n$, define $\bv\in\mathbb{R}^{n-1}$ with
\begin{equation}\label{eq:2}
\bv_i=
\begin{cases}
1 &\mbox{if } \bx_{i+1}<\bx_i \\
-1 &\mbox{if } \bx_{i+1}>\bx_i \\
\in [-1,1] &\mbox{if } \bx_{i+1}=\bx_i.
\end{cases}
\end{equation}
Let $\mathcal{V}$ denote the set of $\bv$'s that satisfy~\eqref{eq:2}, then $\partial f(\bx)$ can be written as %$v\in\mathcal{V}$ if and only if $v$ satisfies \eqref{eq:2}.
\begin{equation}
\partial f(\bx) = \{\bB^{T}\bv:\bv\in\mathcal{V}\}.\label{eq:partialf}
\end{equation}

%Here we restate the weak decomposability for completeness.
\begin{define}
For $\bx\neq \bz$, the set $\partial f(\bx)$ is said to satisfy the \emph{weak decomposability assumption} if there exists $\bw_0\in \partial f(\bx)$ such that
\begin{equation}\label{eq:1}
\dpro{\bw-\bw_0,\bw_0}=0,
\end{equation}
simultaneously for all $ \bw\in\partial f(\bx)$.
\end{define}

Using~\eqref{eq:partialf}, we can rewrite \eqref{eq:1} as
\begin{equation}\label{eq:3}
\exists \bv_0 \in \mathcal{V}\text{ s.t. } \dpro{\bB^{T}\bv-\bB^{T}\bv_0,\bB^{T}\bv_0}=0\text{ },\forall \bv\in\mathcal{V}.
\end{equation}

We have the following result regarding the weak decomposability of $\partial f(\bx)$.
\begin{lem}\label{lem:decompose}
For any given nonzero $\bx\in\mathbb{R}^n$, $\partial f(\bx)$ satisfies the weak decomposability assumption.
\end{lem}
\begin{proof}
To check the decomposability assumption, we need to check whether we can always find a $\bv_0 \in \mathcal{V}$ that satisfies \eqref{eq:3}.

It is easy to check that $\bB\bB^T$ is symmetric, and hence~\eqref{eq:3} is equivalent to
\begin{eqnarray}
%\dpro{B^{T}v-B^{T}v_0,B^{T}v_0}&=&0,\text{ },\forall v\in\mathbb{V},\no\\
%\bv_0^{T}\bB\bB^{T}(\bv-\bv_0)&=&0,\quad\forall \bv\in\mathcal{V},\no\\
\exists \bv_0 \in \mathcal{V}\text{ s.t. }\bv_0^{T}\bB\bB^{T}\bv&=&\bv_0^{T}\bB\bB^{T}\bv_0,\quad\forall \bv\in\mathcal{V}.\label{eq:4}
\end{eqnarray}
\eqref{eq:4} indicates that \eqref{eq:3} is satisfied if and only if we can find a $\bv_0\in\mathcal{V}$ such that $\bv_0^{T}\bB\bB^{T}\bv$ is a constant for all $\bv\in\mathcal{V}$.

Define the set of indices $\cS:=\{i\in[n-1]: \bx_{i}=\bx_{i+1}\}$.
If $\cS=\emptyset$,~\eqref{eq:4} holds trivially, as in this case $\|\bB\bx\|_1$ is differentiable and $\mathcal{V}$ is a singular set. In the following we focus on the case that $\cS\neq\emptyset$.

When $\cS\neq\emptyset$, $\cS$ can be written as a union of consecutive groups of indices that $\cS=\cup_{i=1}^{K+1}[b_i,e_i]$, where $K+1$ is the number of intervals in which the elements in $\bx$ have the same value, $b_i\leq e_i$, $\forall i\in[K+1]$ and $b_{i+1}-e_{i}>1$, $\forall i\in[K]$. $\cS$ can also be expressed explicitly as $\cS = \{\cS_1, \cS_2, \cdots, \cS_{|\cS|}\}$ with elements increasing. We can define $\cS^c$ and $\cS^c$ that have increasing elements in a similar manner.

Using the notation introduced in~\eqref{eq:notation}, we can write $\bv=\bv_{\cS}+\bv_{\cS^c}$, and hence
\begin{equation}\label{eq:5}
\bv_0^{T}\bB\bB^{T}\bv = \bv_0^{T}\bB\bB^{T}\bv_{\cS}+\bv_0^{T}\bB\bB^{T}\bv_{\cS^c}.
\end{equation}
Notice that
\begin{equation}
(\bv_{\cS^c})_i =
\begin{cases}
0, &\mbox{ if }\bx_{i+1}=\bx_i,\\
1, &\mbox{ if } \bx_{i+1}<\bx_i,\\
-1, &\mbox{ if }\bx_{i+1}>\bx_i,\\
\end{cases}\label{eq:sc}
\end{equation}
where $i\in[n-1]$. Given $\bx$, $\bv_{\cS^c}$ is fixed and hence $\bv_0^{T}\bB\bB^{T}\bv_{\cS^c}$ is fixed.
Since $(\bv_{\cS})_i$ can be any real number in $[-1,1]$ for $i\in \cS$, a necessary and sufficient condition for the right hand side of \eqref{eq:5} to be a constant is
\begin{eqnarray}
\bv_0^{T}\bB\bB^{T}\bv_{\cS} &=& 0,\no
\end{eqnarray}
which can be seen by setting $\bv_{\cS}=\bz$. Using notations introduced in~\eqref{eq:subnotation} and~\eqref{eq:matrixnotation}, the equation above can be written as
\begin{eqnarray}
&&\bv_0^{T}(\bB\bB^{T})_{\Omega,\cS}\tilde{\bv}_{\cS} = 0,\quad\forall \bv\in\mathcal{V}\no\\
\Leftrightarrow && \bv_0^{T}(\bB\bB^{T})_{\Omega,\cS}=\bz,\no\\
\Leftrightarrow &&(\bB\bB^{T})_{\cS,\Omega}\bv_0=\bz,\no\\
\Leftrightarrow &&(\bB\bB^{T})_{\cS,\Omega}(\bv_{0})_{\cS}=-(\bB\bB^{T})_{\cS,\Omega}(\bv_{0})_{\cS^c},\no\\
\Leftrightarrow &&(\bB\bB^{T})_{\cS,\cS}(\tilde{\bv}_{0})_{\cS}=-(\bB\bB^{T})_{\cS,\Omega}(\bv_{0})_{\cS^c}.\label{eq:solution}
\end{eqnarray}
If $(\bB\bB^{T})_{\cS,\cS}$ is invertible, from~\eqref{eq:solution}, we obtain
\begin{eqnarray}
(\tilde{\bv}_{0})_{\cS}&=&-((\bB\bB^{T})_{\cS,\cS})^{-1}(\bB\bB^{T})_{\cS,\Omega}(\bv_{0})_{\cS^c}\no\\
&=&-((\bB\bB^{T})_{\cS,\cS})^{-1}(\bB\bB^{T})_{\cS,\cS^c}(\tilde{\bv}_{0})_{\cS^c}.\label{eq:6}
\end{eqnarray}
 Hence, if the answers to the following two questions are both yes:
\begin{enumerate}
\item Is $(\bB\bB^{T})_{\cS,\cS}$ invertible?
\item Is $(\tilde{\bv}_{0})_{\cS}$ produced by \eqref{eq:6} feasible? Or equivalently, does each element of $(\tilde{\bv}_{0})_{\cS}$ fall into the interval $[-1,1]$?
\end{enumerate}
then, combining~\eqref{eq:6} with $(\bv_{0})_{\cS^c}$ in~\eqref{eq:sc}, we find a feasible $\bv_0$ that satisfies the weak decomposability assumption.

To answer the first question, we need to study the structure of $(\bB\bB^{T})_{\cS,\cS}$.
Note that $(\bB\bB^{T})^{-1}$ is symmetric as $\bB\bB^{T}$ shown in \eqref{ddt} is symmetric, and that $(\bB\bB^{T})^{-1}\in\mathbb{R}^{(n-1)\times(n-1)}$. Define $\floor{\cdot}$ and $\ceil{\cdot}$ to be the floor and ceiling operator respectively. Here we give the exact form for $\bB\bB^{T}$ and $(\bB\bB^{T})^{-1}$. For $(\bB\bB^{T})^{-1}$ we only give the upper triangular in \eqref{equ:ddt-1} due to its symmetry.
\begin{figure*}[!t]
% ensure that we have normalsize text
\normalsize
% Store the current equation number.
%\setcounter{MYtempeqncnt}{\value{equation}}
% Set the equation number to one less than the one
% desired for the first equation here.
% The value here will have to changed if equations
% are added or removed prior to the place these
% equations are referenced in the main text.
%\setcounter{equation}{5}
\begin{equation}\label{ddt}
\bB\bB^{T}=\left(
         \begin{array}{ccccccc}
           2 & -1&   &   &   &   &   \\
           -1& 2 & -1&   &   &   &   \\
             & -1& 2 & \ddots &   &   &   \\
             &   & \ddots  &  \ddots &  \ddots &   &   \\
             &   &   & \ddots  &  \ddots &  -1 &   \\
             &   &   &   & -1& 2 & -1\\
             &   &   &   &   & -1& 2 \\
         \end{array}
       \right).\\
\end{equation}
\begin{equation}\label{equ:ddt-1}
(\bB\bB^{T})^{-1}=\frac{1}{n}\left(
         \begin{array}{ccccccc}
           (n-1) & (n-2)& \cdots  &  \ceil{\frac{n}{2}} & \cdots  & 2  & 1  \\
             & 2(n-2) & &  2\ceil{\frac{n}{2}} &   &  2\times2 &  2 \\
                 &        &  \cdots  & &   &   &  \cdots \\
             &   &   & \floor{\frac{n}{2}}\ceil{\frac{n}{2}}  &   &   & \floor{\frac{n}{2}}  \\
             &   &   &   & \cdots &   &  \cdots \\
             &   &   &   & & 2(n-2) & (n-2)\\
             &   &   &   &   & & (n-1)\\
         \end{array}
       \right).
\end{equation}
\begin{equation}\label{ddtgimel}
(\bB\bB^{T})_{\cS,\cS}=\left(
         \begin{array}{ccccccc}
           2 & -\delta_{\cS_2-\cS_1=1}&   &   &   &   &   \\
           -\delta_{\cS_2-\cS_1=1}& 2 & -\delta_{\cS_3-\cS_2=1}&   &   &   &   \\
             & -\delta_{\cS_3-\cS_2=1}& 2 & \ddots &   &   &   \\
             &   & \ddots  &  \ddots & -\delta_{\cS_{K+1}-\cS_{K}=1}  &   &   \\
             &   &   &  -\delta_{\cS_{K+1}-\cS_{K}=1} &  2 & &  \\
         \end{array}
       \right).
\end{equation}
% Restore the current equation number.
%\setcounter{equation}{\value{MYtempeqncnt}}
% IEEE uses as a separator
\hrulefill
% The spacer can be tweaked to stop underfull vboxes.
\vspace*{4pt}
\end{figure*}
\begin{eqnarray}
(\bB\bB^{T})_{i,j}&=&
\begin{cases}
2, &\mbox{ if } i=j,\\
-1, &\mbox{ if }|i-j|=1,\\
0, &\mbox{ otherwise}.
\end{cases}\\
\left((\bB\bB^{T})^{-1}\right)_{i,j}&=&
\begin{cases}
\frac{i(n-j)}{n}, &\mbox{ if } i\leq j,\\
\frac{j(n-i)}{n}, &\mbox{ if } i>j.
\end{cases}
\end{eqnarray}
From \eqref{ddtgimel}, we have
\begin{eqnarray}
&&(\bB\bB^{T})_{\cS,\cS,i,j}\no\\
&=&\sum_{p=1}^n (\bB_{\cS,})_{i,p}(\bB_{\cS,})_{j,p}\no\\
&=&
\begin{cases}
2, &\mbox{ if } i=j,\\
-\delta_{|\cS_i-\cS_j|=1}, &\mbox{ if }|i-j|=1,\\
0, &\mbox{ otherwise},
\end{cases}
\end{eqnarray}
in which $\delta_{|\cS_i-\cS_j|=1}$ is the indicator function.

Recall that $\cS=\cup_{i=1}^{K+1}[b_i,e_i]$, where $b_i\leq e_i$, $\forall i\in[K+1]$ and $b_{i+1}-e_{i}>1$, $\forall i\in[K]$. Let $I_i:=e_i-b_i+1$ denote the length of $i$th group. Hence $|\cS|=\sum_{i=1}^{K+1}I_i$.
For a positive integer $l$, define matrix $\bH(l)\in\mathbb{R}^{l\times l}$
\begin{eqnarray}
\left((\bH(l)\right)_{i,j}&=&
\begin{cases}
2, &\mbox{ if } i=j,\\
-1, &\mbox{ if }|i-j|=1,\\
0, &\mbox{ otherwise}.
\end{cases}
\end{eqnarray}
So $(\bB\bB^{T})_{\cS,\cS}$ can be expressed as
\begin{eqnarray}
&&(\bB\bB^{T})_{\cS,\cS} \no\\
&=& \left(
               \begin{array}{cccc}
                 \bH(I_1) &  &  &  \\
                  & \bH(I_2) &  &  \\
                  &  & \ddots &  \\
                  &  &  & \bH(I_{K+1}) \\
               \end{array}
             \right),\no
\end{eqnarray}
and
\begin{eqnarray}
&&((\bB\bB^{T})_{\cS,\cS})^{-1} \no\\
&=& \left(
               \begin{array}{cccc}
                 \bH(I_1)^{-1} &  &  &  \\
                  & \bH(I_2)^{-1} &  &  \\
                  &  & \ddots &  \\
                  &  &  & \bH(I_{K+1})^{-1} \\
               \end{array}
             \right),\no\\\label{eq:7}
\end{eqnarray}
where
\begin{eqnarray}
\left((\bH(l))^{-1}\right)_{i,j}&=&
\begin{cases}
\frac{i(l-j)}{l}, &\mbox{ if } i\leq j,\\
\frac{j(l-i)}{l}, &\mbox{ if } i>j.
\end{cases}\label{eq:8}
\end{eqnarray}

\eqref{eq:7} implies that the answer to the first question is yes. Now, we investigate the second question. For that, we first study the structure of $(\bB\bB^{T})_{\cS,\cS^c}$.
Note that $\cS\cap\cS^c=\emptyset$,
\begin{equation}
\left((\bB\bB^{T})_{\cS,\cS^c}\right)_{i,j}=
\begin{cases}
-1, &\mbox{ if }|\cS_i-(\cS^c)_j|=1,\\
0, &\mbox{ otherwise}.
\end{cases}\label{eq:9}
\end{equation}
Notice that $((\bB\bB^{T})_{\cS,\cS})^{-1}$ is in block form, we can also divide $(\bB\bB^{T})_{\cS,\cS^c}$ into blocks corresponding to $((\bB\bB^{T})_{\cS,\cS})^{-1}$.
\begin{equation}
(\bB\bB^{T})_{\cS,\cS^c}
=\left(
  \begin{array}{c}
    (\bB\bB^{T})_{\cS,\cS^c}^{(I_1)} \\
    \cdots \\
    (\bB\bB^{T})_{\cS,\cS^c}^{(I_{K+1})} \\
  \end{array}
\right),
\end{equation}
where $(\bB\bB^{T})_{\cS,\cS^c}^{(I_i)}\in\mathbb{R}^{I_i\times |\cS^c|}$ denote the $i$th block.

Now we have
\begin{equation}
((\bB\bB^{T})_{\cS,\cS})^{-1}(\bB\bB^{T})_{\cS,\cS^c}
=\left(
  \begin{array}{c}
    \bH(I_1)^{-1}(\bB\bB^{T})_{\cS,\cS^c}^{(I_1)} \\
    \cdots \\
    \bH(I_{K+1})^{-1}(\bB\bB^{T})_{\cS,\cS^c}^{(I_{K+1})} \\
  \end{array}
\right).
\end{equation}

Next, we conduct a more close analysis of $(\bB\bB^{T})_{\cS,\cS^c}^{(I_i)}$, $\forall i\in[K+1]$.
Note that the interval with length $I_i$ corresponds to indices $[b_i,e_i]$ of $\bx$, due to condition in \eqref{eq:8}, $-1$ can only appear at position $(j,l)$ when $\cS_j= b_i$ and $\cS^c_l= b_i\pm 1$, or when $\cS_j= e_i$ and $\cS^c_l= e_i\pm 1$. Now, consider two cases:

\emph{Case 1:} If $b_i=e_i$, then $b_i+1=e_i+1$ and $b_i-1=e_i-1$. So $-1$ can only appear at most two positions, the resulting row vector $\bH(I_i)^{-1}(\bB\bB^{T})_{\cS,\cS^c}^{(I_i)}$ has at most two nonzero elements which are equal to $-\frac{1}{2}$ due to $\bH(I_{i})^{-1}=\frac{1}{2}$.

\emph{Case 2:} If $b_i\neq e_i$, then $b_i+1\in[b_i,e_i]\notin \cS^c$ and $e_i-1\in[b_i,e_i]\notin \cS^c$. So $-1$ can only appear at most two positions, which we know must lie in the first row and last row of $(\bB\bB^{T})_{\cS,\cS^c}^{(I_i)}$ respectively, since the points in $(b_i,e_i)$ have no points in $\cS^c$. The first element and last element in each row, say $l$, of $\bH(I_{i})^{-1}$ are $-\frac{I_i-l}{I_i}$ and $-\frac{l}{I_i}$ from \eqref{eq:9}. Hence each row $l$ in the result matrix $\bH(I_i)^{-1}(\bB\bB^{T})_{\cS,\cS^c}^{(I_i)}$ has at most two nonzero elements which are $-\frac{I_i-l}{I_i}$ and $-\frac{l}{I_i}$. Note that $\frac{I_i-l}{I_i}+\frac{l}{I_i}=1$ and $-1\leq\frac{I_i-l}{I_i},\frac{l}{I_i}\leq 1$.

Combining these two cases, we know that each row in $((\bB\bB^{T})_{\cS,\cS})^{-1}(\bB\bB^{T})_{\cS,\cS^c}$ has at most two nonzero elements which falls between $[-1,1]$ and whose sum is $-1$. Since each element in $(\tilde{\bv}_{0})_{\cS^c}$ falls into $[-1,1]$, the resulting $(\tilde{\bv}_{0})_{\cS}$ is always feasible. This implies that the answer to the second question is also yes.

As the result, we find a $\bv_0$, by combining~\eqref{eq:sc} and~\eqref{eq:6}, that satisfies the weak decomposability. The proof of the lemma is complete.
\end{proof}
%\section{Proof for $\min_{\lambda\geq 0}D(\lambda\partial f(x))\approx D(cone(\partial f(x)))$}\label{sec:proof}

With Lemma~\ref{lem:decompose}, we are ready to state the main result.
\begin{thm}
The phase transition of the TV minimization problem is $\min_{\lambda\geq 0}D(\lambda\partial f(\bx))$.
\end{thm}
\begin{proof}
We will use Proposition 1 in \cite{Foygel:TIT:2014}, which also applies to any other convex complexity measure.
As Lemma 1 shows that $\partial f(\bx)$ satisfies the weak decomposability, using Proposition 1 in \cite{Foygel:TIT:2014}, we have
\begin{eqnarray}
\min_{\lambda\geq 0}D(\lambda\partial f(\bx))&=&\min_{\lambda\geq 0}\mathbb{E}\left\{\inf_{\bu\in\lambda\partial f(\bx)}\|\bg-\bu\|_2^2\right\}\no\\
&\leq&\mathbb{E}\left\{\min_{\lambda\geq 0}\inf_{\bu\in \lambda\partial f(\bx)}\|\bg-\bu\|_2^2\right\}+6\no\\
&=&D(\text{cone}(\partial f(\bx)))+6.\label{eq:16}
\end{eqnarray}
We also have
\begin{eqnarray}
\min_{\lambda\geq 0}D(\lambda\partial f(\bx))&=&\min_{\lambda\geq 0}\mathbb{E}\left\{\inf_{u\in\lambda\partial f(\bx)}\|\bg-\bu\|_2^2\right\}\no\\
&\geq&\mathbb{E}\left\{\min_{\lambda\geq 0}\inf_{\bu\in \lambda\partial f(\bx)}\|\bg-\bu\|_2^2\right\}\no\\
&=&D(\text{cone}(\partial f(\bx))).\label{eq:17}
\end{eqnarray}
Combining \eqref{eq:16} and \eqref{eq:17}, we have
\begin{equation}
D(\text{cone}(\partial f(\bx)))\leq\min_{\lambda\geq 0}D(\lambda\partial f(\bx))\leq D(\text{cone}(\partial f(\bx)))+6.
\end{equation}
Since $\min_{\lambda\geq 0}D(\lambda\partial f(\bx))$ grows proportionally with $n$ when the sparsity of the gradient grow proportionally with $n$, as shown in \cite{Donoho:TIT:2013},
the approximation error $6$ is negeligible. Thus we complete our proof.
\end{proof}
\section{Conclusion}\label{sec:conclusion}
We have verified that the TV regularizer satisfies the weak decomposability condition. We have proved $\min_{\lambda\geq 0}D(\lambda\partial f(\bx))\approx D(\text{cone}(\partial f(\bx)))$ for the TV regularizer $f(\bx)$. Thus the minmax MSE result derived in Donoho's paper\cite{Donoho:TIT:2013} directly gives the phase transition of the total variation minimization.
\section*{Acknowledgement}
We thank Professor Joel Tropp for helpful discussions on the approximation on Gaussian width when Bingwen Zhang and Weiyu Xu visited IMA at University of Minnesota. 
\bibliographystyle{IEEEtran}
\bibliography{macros,messagePassing}
\end{document}